\documentclass[9pt]{entcs}

\usepackage{entcsmacro}

\usepackage{amsmath}
\usepackage{graphics}
\usepackage{stmaryrd}
\usepackage{bussproofs}
\usepackage[all,2cell]{xy}
\usepackage{xcolor}
\usepackage{multicol}
\usepackage{wrapfig}

\EnableBpAbbreviations

\UseTwocells

\newcommand*{\lambdabox}{$\lambda^\Box$}
\newcommand*{\lambdabracket}{$\lambda^{[-]}$}
\newcommand*{\Circ}{\bigcirc}
\newcommand*{\lambdacircle}{$\lambda^\Circ$}
\newcommand*{\IK}{$\mathrm{IK}$}

\newcommand*{\N}{\mathbb{N}}
\newcommand*{\tuple}[1]{{\langle{#1}\rangle}}
\newcommand*{\den}[1]{{\llbracket{#1}\rrbracket}}
\newcommand*{\iso}{\mathrel{\cong}}
\newcommand*{\after}{\mathbin{\circ}}
\newcommand*{\tensor}{\mathbin{\otimes}}
\newcommand*{\quo}{\lq}
\newcommand*{\unq}{{,}}
\newcommand*{\rg}{\operatorname{rg}}
\newcommand*{\dom}{\operatorname{dom}}
\newcommand*{\ob}{\operatorname{ob}}

\newcommand*{\ev}{\operatorname{ev}}
\newcommand*{\A}{\mathcal{A}}
\newcommand*{\B}{\mathcal{B}}
\newcommand*{\E}{\mathcal{E}}
\newcommand*{\V}{\mathcal{V}}
\newcommand*{\W}{\mathcal{W}}
\newcommand*{\VV}{\mathfrak{V}}
\newcommand*{\CC}{\mathfrak{C}}

\newcommand*{\PSh}[1]{\widehat{#1}}
\newcommand*{\underlying}[1]{{#1}_0}
\newcommand*{\catname}[1]{{\normalfont\textbf{#1}}}
\newcommand*{\Cat}{\catname{Cat}}
\newcommand*{\twoCat}{\catname{2Cat}}
\newcommand*{\MonCat}{\catname{MonCat}}

\newcommand*{\Set}{\catname{Set}}
\newcommand*{\VCat}{\V\mbox{-}\catname{Cat}}
\newcommand*{\WCat}{\W\mbox{-}\catname{Cat}}
\newcommand*{\MontwoCat}{\catname{Mon2Cat}}
\newcommand*{\SMCat}{\catname{SMCat}}
\newcommand*{\SMCCat}{\catname{SMCCat}}
\newcommand*{\SMtwoCat}{\catname{SM2Cat}}

\newcommand*{\Mon}{\operatorname{Mon}}
\newcommand*{\SM}{\operatorname{SM}}
\newcommand*{\MonVCat}{\V\mbox{-}\catname{MonCat}}

\newcommand*{\SMVCat}{\V\mbox{-}\catname{SMCat}}
\newcommand*{\SMCVCat}{\V\mbox{-}\catname{SMCCat}}
\newcommand*{\NorSMCCat}{\catname{NorSMCCat}}

\newcommand*{\PInj}{\catname{PInj}}

\def\lastname{Nishiwaki, Kakutani, \and Murase}

\begin{document}

\begin{frontmatter}
    \title{Modality via Iterated Enrichment}
    \author{Yuichi Nishiwaki\thanksref{email1}}
    \address{Department of Computer Science, University of Tokyo}
    \author{Yoshihiko Kakutani\thanksref{email2}}
    \address{Department of Computer Science, University of Tokyo}
    \author{Yuito Murase\thanksref{email3}}
    \address{Department of Computer Science, University of Tokyo}
    \thanks[email1]{Email: \href{mailto:nyuichi@is.s.u-tokyo.ac.jp} {\texttt{\normalshape nyuichi@is.s.u-tokyo.ac.jp}}}
    \thanks[email2]{Email: \href{mailto:kakutani@is.s.u-tokyo.ac.jp} {\texttt{\normalshape kakutani@is.s.u-tokyo.ac.jp}}}
    \thanks[email3]{Email: \href{mailto:murase@lyon.is.s.u-tokyo.ac.jp} {\texttt{\normalshape murase@lyon.is.s.u-tokyo.ac.jp}}}
    \begin{abstract}
        This paper investigates modal type theories by using a new categorical semantics called change-of-base semantics.
        Change-of-base semantics is novel in that it is based on (possibly infinitely) iterated enrichment and interpretation of modality as hom objects.
        In our semantics, the relationship between meta and object levels in multi-staged computation exactly corresponds to the relationship between enriching and enriched categories.
        As a result, we obtain a categorical explanation of situations where meta and object logics may be completely different.
        Our categorical models include conventional models of modal type theory (e.g., cartesian closed categories with a monoidal endofunctor) as special cases and hence can be seen as a natural refinement of former results.
    
        On the type theoretical side, it is shown that Fitch-style modal type theory can be directly interpreted in iterated enrichment of categories.
        Interestingly, this interpretation suggests the fact that Fitch-style modal type theory is the right adjoint of dual-context calculus.
        In addition, we present how linear temporal, S4, and linear exponential modalities are described in terms of change-of-base semantics.
        Finally, we show that the change-of-base semantics can be naturally extended to multi-staged effectful computation and generalized contextual modality a la Nanevski et al.
        We emphasize that this paper answers the question raised in the survey paper by de Paiva and Ritter in 2011, what a categorical model for Fitch-style type theory is like.
    \end{abstract}
\end{frontmatter}

\section{Introduction}

Since the Curry-Howard isomorphism was proposed, intuitionistic logic has attracted many logicians and computer scientists.
As a result, the basic results not only on propositional implicational logic but also on dependent or higher-order extensions were established.
Nevertheless, how to deal with (necessity) modality in intuitionistic logic still remains to be an intricate problem.
Especially the type theoretic aspects (or equivalently, natural deduction systems) of intuitionistic modal logic are rather undeveloped.

For the logical aspects of intuitionistic modality, the currently accepted form of definition first appeared independently in \cite{Fischer84,DBLP:conf/tark/PlotkinS86}.
In their papers, Kripke-style semantics called birelational models is defined, and Hilbert-style axiomatizations for some intuitionistic modal logics including K, S4, and S5 are provided.
The history of intuitionistic modal logic and systematic comparisons of various systems ever proposed can be found in \cite{DBLP:phd/ethos/Simpson94a,Kojima12}.
For the type theoretic aspects, on the other hand, the approaches are diverse.
As of this writing, there are mainly three types of natural deduction systems that have gained popularity, called \emph{Gentzen-style} \cite{Bellin01,DBLP:conf/aplas/Kakutani07}, \emph{dual-context} \cite{DBLP:journals/entcs/PaivaR16,DBLP:conf/lics/Kavvos17}, and \emph{Fitch-style} \cite{Martini96,DBLP:journals/corr/abs-1710-08326} systems.
For the first two systems, their computational and categorical aspects are intensively investigated in a number of papers \cite{Paiva11,DBLP:journals/corr/Kavvos16c}, and applied in multi-staged computation \cite{DBLP:journals/jacm/Davies17,DBLP:journals/tocl/NanevskiPP08}.
However, despite that it was temporally the first system of modal type theory for intuitionistic K, Fitch-style type theory yet remains the least developed regarding its operational and categorical semantics.
In a survey paper \cite{Paiva11}, de Paiva and Ritter asked how one can describe Fitch-style type theory computationally and categorically, not in terms of syntactic translation via Gentzen-style calculus.
Clouston, in a recent paper \cite{DBLP:journals/corr/abs-1710-08326}, gave a partial answer to this question by proposing a sound but incomplete categorical model for Fitch-style modal calculus.
The main aim of this paper is to answer fully this question.
Our slogan is ``\textit{Fitch-style type theory corresponds to iterated enrichment of categories, and boxes are hom objects}''.

In this paper, we offer a categorical semantics to Fitch-style modal type theory from the viewpoint of enriched categories.
In our semantics, if a type $A$ is interpreted as an object in an enriched category $\A$, $\Box A$ is then interpreted as an object in its enriching category $\V$.
Similarly, judgments before and after necessitation are interpreted as morphisms in (the underlying category of) $\A$ and $\V$, respectively.
Following this style, it would be natural to distinguish judgments with and without boxed types also in syntax.
This kind of distinction is quite reasonable especially when we regard the calculus as multi-staged computation, as compile-time and run-time environments are usually different.
For that reason, we first generalize Fitch-style modal type theory by introducing \emph{levels} of judgments.
Just by ignoring levels, we obtain a usual Fitch-style calculus.
In addition to the syntactic ingredient, we also need to formalize \emph{infinitely iterated enrichment} of categories, since modality may be nested arbitrarily deeply.
In this paper, we exploit change-of-base construction of enriched categories to define such a structure.

Our research contributions are summarized as follows.

\begin{itemize}
    \item We discovered that modality in type theory captures the enrichment structure in the categorical semantics.
    In particular, we pointed out that (contextual) modality can be viewed as \emph{external hom objects}.
    \item We defined iterated enrichment based upon the change-of-base construction and introduced two constructions of them.
    \item We compared our semantics syntactically e.g. with dual-context calculi and semantically e.g. with linear non-linear models.
\end{itemize}

\emph{Organization.}
In Section~\ref{sec:typetheory}, we introduce multi-level Fitch-style modal type theory {\lambdabox}, the type theory in question in this paper.
We also compare it with other styles of modal type theories.
In Section~\ref{sec:semantics}, we show that infinitely iterated enrichment of categories gives sound and complete semantics of {\lambdabox}.
In Section~\ref{sec:constructions}, we present two constructions of infinitely iterated enriched categories.
In Section~\ref{sec:models}, we explain how one can model various kinds of computation in terms of our semantics.
Particularly, program-generating programs, $!$-modality in linear calculus, and multi-staged effectful computation are discussed.
In Section~\ref{sec:contextual}, we generalize the type theory and the semantics to contextual modality.

\section{Modal Type Theories}\label{sec:typetheory}

We introduce the multi-level Fitch-style modal type theory {\lambdabox}.
We also compare {\lambdabox} with other variants of modal type theories.

\subsection{Multi-level Fitch-style System}

\begin{figure}[t]
    \begin{multicols}{2}
        \begin{description}
            \item[Levels] $l \in \N$
            \item[Types] $A^{l}, B^{l} ::= \iota^{l} \mid A^{l} \to B^{l} \mid \Box A^{l-1} \ (l \ge 1)$
        \end{description}
        \columnbreak
        \begin{description}
            \item[Contexts] $\Gamma^{l} ::= x_1 : A^{l}_1,\cdots, x_n : A^{l}_n$
            \item[Context Stacks] $\Delta^{l} ::= \Gamma^{l+n-1};\cdots; \Gamma^{l}$
        \end{description}
    \end{multicols}
    \vspace{-10pt}
    \[
        \fbox{$\Delta^{l} \vdash^l M : A^l$}
        \quad
        \AXC{$(x : A^l) \in \Gamma^l$}
        \RL{Var}
        \UIC{$\Delta^{l+1}; \Gamma^l \vdash^l x : A^l$}
        \DisplayProof
        \quad
        \AXC{$\Delta^{l+1}; \Gamma^l, x : A^l \vdash^l M : B^l$}
        \RL{Abs}
        \UIC{$\Delta^{l+1}; \Gamma^l \vdash^l \lambda x.M : A^l \to B^l$}
        \DisplayProof
    \]
    \[
        \AXC{$\Delta^l \vdash^l M : A^l \to B^l$}
        \AXC{$\Delta^l \vdash^l N : A^l$}
        \RL{App}
        \BIC{$\Delta^l \vdash^l MN : B^l$}
        \DisplayProof
        \quad
        \AXC{$\Delta^{l+1}; \cdot \vdash^l M : A^l$}
        \RL{Quo}
        \UIC{$\Delta^{l+1} \vdash^{l+1} \quo M : \Box A^l$}
        \DisplayProof
        \quad
        \AXC{$\Delta^{l+1} \vdash^{l+1} M : \Box A^l$}
        \RL{Unq}
        \UIC{$\Delta^{l+1}; \Gamma^l \vdash^l \unq M : A^l$}
        \DisplayProof
    \]
    \caption{Inference rules for {\lambdabox}}
    \label{fig:lambdabox}
\end{figure}
Figure~\ref{fig:lambdabox} presents the inference rules for the {\lambdabox} calculus which we work on in this paper.
Our {\lambdabox} is a Fitch-style modal natural deduction system\cite{DBLP:journals/jolli/Borghuis98}.
A (hypothetical) judgment in a Fitch-style modal type theory differs from ones in ordinary type theories in that it has on the left side a stack of contexts instead of a single context.
As shown in Figure~\ref{fig:lambdabox}, contexts in a context stacks are delimited with semicolons in the literature.
Most of the time, only the rightmost compartment of a context stack is concerned with type derivation.
Only when introducing a box or unboxing a term, the context stack is (un)shifted and the second rightmost compartment is then focused.
Boxing is only allowed for judgments with the rightmost compartment empty.
This corresponds to the fact that a boxed proposition is an assertion of the validity of the proposition in modal logic.
For the more detailed description of Fitch-style calculi, a survey \cite{Paiva11} should be helpful.
Notations of boxing and unboxing vary with systems: they are denoted by $\mathtt{box}(M)$ and $\mathtt{unbox}(M)$ in \cite{DBLP:journals/sLogica/BiermanP00,DBLP:conf/popl/KimYC06}, $\mathtt{gen}M$ and $\mathtt{ungen}M$ in \cite{Martini96}, $\hat{k}{M}$ and $\check{k}{M}$ in \cite{DBLP:journals/jolli/Borghuis98}, \texttt{<$M$>} and \texttt{\textasciitilde$M$} in MetaML \cite{DBLP:journals/tcs/TahaS00}, $\mathtt{shut}M$ and $\mathtt{open}M$ in \cite{DBLP:journals/corr/abs-1710-08326}, and $\mathtt{next}M$ and $\mathtt{prev}M$ in {\lambdacircle} \cite{DBLP:journals/jacm/Davies17}.
We borrowed terms \emph{quotation} and \emph{unquotion} and notations $\quo M$ or $\unq M$ from Lisp \cite{DBLP:books/lib/Steele90}.

Our {\lambdabox}, however, differs from Borghuis' original formulation in some points.
One difference is that we omitted rules and constructs for higher-order and dependent functionalities since we are only interested in the first-order propositional fragment.
However, removal of the weakening rule from the calculus makes the deduction theorem fail.
To remedy this issue, we modified unquotation to allow the rightmost context to be weakened\footnote{In \cite{Paiva11}, Borghuis' calculus is similarly reduced to the propositional fragment, but this issue is left unfixed.}, which approach is also taken in \cite{DBLP:journals/synthese/HakliN12}.
Another difference is that, in {\lambdabox}, every type and judgment has an intrinsic level which is a natural number.
A level may be viewed as the \emph{stage} of the computation.
Whereas abstraction and application are closed under each level, quotation and unquotation change the level, i.e., the stage of computation by one.
Since all base types are leveled by default, our {\lambdabox} models computation in which there may be different universes of types and primitives for each stage.
In the rest of the paper, we always use the turnstyle for leveled systems by $\vdash^l$ and that for unleveled systems by $\vdash$.
However, levels of types and contexts in a judgment are often omitted because they are easily inferable from the level of the judgment.
We call the corresponding logic of the {\lambdabox} calculus {\IK}.

\begin{lemma}
    The weakening, exchange, contraction, and substitution rules are admissible in {\IK}.
\end{lemma}

The exchange, contraction, and substitution rules must be operated for each level.
For example, the substitution rule is explained as follows.
\[
    \AXC{$\Delta; \Gamma_1, A, \Gamma_2; \Delta' \vdash^{l} B$}
    \AXC{$\Delta; \Gamma_3 \vdash^{l+\#\Delta'} A$}
    \RL{S}
    \BIC{$\Delta; \Gamma_1, \Gamma_3, \Gamma_2; \Delta' \vdash^{l} B$}
    \DisplayProof
\]
Note that $\#\Delta$ denotes the depth of $\Delta$ as a stack.
The following theorem renders the name of the logic consistent with the convention.
There is an obvious forgetful function $|\cdot|$ from leveled judgments to unleveled judgments that forgets levels.

\begin{theorem}\label{thm:labeling}
    $\Delta \vdash A$ if and only if there exists a level $l$ such that $\Delta \vdash^l A$ and $|\Delta \vdash^l A| = \Delta \vdash A$.
\end{theorem}

It is known that unleveled Fitch-style system is equivalent to the smallest intuitionistic normal modal logic w.r.t.\ provability.
The leveled Fitch-style system hence just refines proofs of intuitionistic K by introducing levels.
Note that $|\cdot|$ can be easily extended to a proof-relevant function (e.g., $|\lambda x.M : A^l \to B^l| = \lambda x_l.|M : B^l|$).
With this extension, the only-if direction in the above theorem states that given a proof tree in the unleveled Fitch-style modal calculus, we can reconstruct a proof in the leveled system just by attaching levels assuming appropriate base types.
For example, a proof term for the leveled version of the axiom K is $\lambda x. \lambda y. \quo ((\unq x)(\unq y)) : \Box (A^l \to B^l) \to \Box A^l \to \Box B^l$.
By the following lemma, we can identify two context stacks $\Delta$ and $\cdot; \Delta$, where $\cdot$ is an empty context.

\begin{lemma}
    $\cdot; \Delta \vdash^l A$ if and only if $\Delta \vdash^l A$.
\end{lemma}

Then the so-called \emph{denecessitation theorem} \cite[Chapter 20]{Jackson05} immediately follows.
We will later review this in Section \ref{sec:semantics}.

\begin{corollary}[Denecessitation]
    If $\vdash^{l+1} \Box A$, then $\vdash^l A$.
\end{corollary}

Computational behaviors of quotation and unquotation are explained in terms of code generation: quotation creates a code template and unquotation makes holes in a template.
Formally, the dynamics of {\lambdabox} are defined in terms of the following rules for the $\Box$ type in addition to the usual $\beta\eta$ rules.
\begin{align*}
    \unq\quo M &\xrightarrow{\beta_\Box} M & M &\xrightarrow{\eta_\Box} \quo\unq M \text{ if $M : \Box A$}
\end{align*}
Although the $\beta\eta$ rules for the $\Box$ type look intuitive enough, the $\beta$ rule for the $\to$ type must be carefully extended to terms with mixed levels.

\begin{definition}
    Substitution $[N/x]M$ is given by $[N/x]_0M$, where $[N/x]_n$ is inductively defined as follows.
    \begin{align*}
        [N/x]_n{y} &= \begin{cases}
            N & (n = 0 \land x = y) \\
            y & (\text{otherwise})
        \end{cases} &
        [N/x]_n \quo M &= \quo [N/x]_{n+1} M &
        [N/x]_n \unq M &= \begin{cases}
            \unq [N/x]_{n-1} M & (n > 0) \\
            \unq M & (\text{otherwise})
        \end{cases}
    \end{align*}
    Other cases are omitted.
\end{definition}

Our substitution rule is intended to allow two contexts of different levels to contain variables with the same literal names.
Consider a term $(\lambda\textcolor{blue}{x}.\quo\lambda\textcolor{red}{x}.(\textcolor{red}{x}(\unq\textcolor{blue}{x})))y$.
The first occurrence of $\textcolor{red}{x}$ is a constituent of the code template $\quo\lambda\textcolor{red}{x}.(\textcolor{red}{x}(\unq{-}))$ whereas the second one $\textcolor{blue}{x}$ is not (notice that they have different levels), so $y$ is substituted only for the second occurrence, resulting in $\quo\lambda\textcolor{red}{x}.(\textcolor{red}{x}(\unq y))$.
{\lambdabox} enjoys the usual meta-theoretical properties.

\begin{theorem}
    Subject reduction, strong normalization, and the Church-Rosser property hold for the $\beta$ rules.
\end{theorem}

\subsection{Other Systems}\label{sec:other}

Many deductive systems of modal logic other than Fitch-style have been proposed.
Especially, streams of Gentzen-style and dual-context systems are important in the field of computer science.

The Gentzen-style system is basically a usual natural deduction system of intuitionistic propositional logic but together with the following additional inference rule for boxes.
In Gentzen-style, in contrast to Fitch-style, the introduction and elimination rules for $\Box$ are mixed into a single rule; logical harmony a la Dummett is dismissed.
\[
    \AXC{$x_1 : A_1,\cdots, x_n : A_n \vdash M : B$}
    \AXC{$\Gamma \vdash N_1 : \Box A_1 \quad \cdots \quad \Gamma \vdash N_n : \Box A_n$}
    \BIC{$\Gamma \vdash \mathop{\mathtt{box}} x_1,\cdots,x_n \mathrel{\mathtt{be}} N_1,\cdots,N_n \mathrel{\mathtt{in}} M : \Box B$}
    \DisplayProof
\]
This calculus is investigated in some papers \cite{Bellin01,DBLP:conf/aplas/Kakutani07,DBLP:journals/corr/Kakutani16}.
Logical provability of the Gentzen-style modal type theory is equivalent to intuitionistic K.
The desired syntactic properties such as strong normalization and categorical semantics are provided in the papers.
It has been shown that the calculus is sound and complete for cartesian closed categories with lax monoidal endofunctors.
As discussed in \cite{DBLP:journals/corr/Kakutani16}, it is also possible to consider strong monoidal functors with some additional equations.
Here, we refer to the strong monoidal version of the soundness and completeness.

\begin{theorem}
    A \emph{Kripke category} is a cartesian closed category endowed with a strong monoidal (i.e., finite product-preserving) endofunctor.
    Kripke categories are sound and complete for the Gentzen-style modal type theory.
\end{theorem}

Dual-context system was first introduced in \cite{Barber96} to study exponentials of linear logic (IMELL), and later refined by several authors \cite{DBLP:journals/sLogica/BiermanP00,DBLP:conf/lics/Kavvos17} to accommodate more logics including the intuitionistic K and S4.
Judgments in dual-context systems have two contexts separated by $\mid$, the left side of which is called \emph{modal context}.
As well as Fitch-style, only the right-hand side context is used when deriving non-modal constructs.
The box of the dual-context system is characterized by the following two rules.
\[
    \AXC{$\cdot \mid \Gamma \vdash A$}
    \RL{$\Box$-I}
    \UIC{$\Gamma \mid \Gamma' \vdash \Box A$}
    \DisplayProof
    \qquad
    \AXC{$\Gamma \mid \Gamma' \vdash \Box A$}
    \AXC{$\Gamma, A \mid \Gamma' \vdash B$}
    \RL{$\Box$-E}
    \BIC{$\Gamma \mid \Gamma' \vdash B$}
    \DisplayProof
\]
It is also known that a Kripke category can be a model of the dual-context system \cite{DBLP:conf/lics/Kavvos17}.
The interpretation of a dual-context judgment in a Kripke category with $F$ is straightforward if we assign terms to proofs appropriately.
\[
    \den{\Gamma \mid \Gamma' \vdash M : A} : F\den{\Gamma} \times \den{\Gamma'} \to \den{A}
\]
We shall revisit the dual-context system in Section~\ref{sec:linearnonlinear}, and discuss another kind of semantics.

\section{Change-of-base Semantics}\label{sec:semantics}

This section presents our main results, soundness and completeness results of the change-of-base semantics.

\subsection{Idea}\label{sec:idea}

To illustrate our idea briefly, we give an interpretation of the two-level fragment of {\lambdabox} in the change-of-base semantics.
In the following, we assume that for any type judgment all contexts whose level is greater than one are trivial.
We also assume that the quotation and unquotation rules in level $1$ are disabled.

The categorical setting is as follows.
Suppose $\V$ is a cartesian closed category and $\A$ is a cartesian closed $\V$-enriched category, which is a $\V$-category with $\V$-enriched (finite) products and exponentials (defined as the right $\V$-adjoint), with $\V$-natural transformations $\pi : A \times B \Rightarrow A$, $\epsilon : B^A \times A \Rightarrow B$, and $! : A \Rightarrow 1$.
Any judgments of both level $0$ and $1$ are interpreted as morphisms in $\V$.
But types are interpreted as objects either of $\A$ and $\V$ depending on their levels: types of level $0$ become objects in $\A$ and ones of level $1$ in $\V$.

We define an interpretation of judgments by induction on typing derivation.
For readability, we may identify terms with judgments and write just $\den{M}$ for $\den{\Delta \vdash^l M : A}$.
If $M$ is of level $1$ and involves no box types, we interpret it in $\V$ in the conventional way of interpretation of $\lambda^\to$ in cartesian closed categories \cite{Lambek86}.
If $M$ is of level $0$ and its context of level $1$ is trivial, we interpret it in the underlying category $\underlying{\A}$ of $\A$, which is in fact a cartesian closed (ordinary) category, in the usual way again.
Note that a morphism in $\underlying{\A}$ is a morphism in $\V$ by definition: $X \to Y \text{ in $\underlying{\A}$} \sslash 1 \to \A(X,Y) \text{ in $\V$}$.
Therefore, $\den{\Gamma \vdash^0 M : X} : \den{\Gamma} \to \den{X} = 1 \to \A(\den{\Gamma}, \den{X})$.
Otherwise, we interpret it directly in $\V$ as in Figure~\ref{fig:twolevelinterp}.
$\Box$ is interpreted as the monoidal covariant hom functor $\A(1, -)$.
In this way, judgments $\Gamma; \cdot \vdash A$ and $\Gamma \vdash \Box A$ become exactly the same morphism.
We consider this reflects the fact that quotation involves no computational contents and is just a transition of viewpoints on the same morphism, from object level to meta level.

\begin{figure*}[t]
    \begin{center}
        \AXC{$\Gamma';\Gamma \vdash^0 x : A$}
        \RL{Var}
        \UIC{$\Gamma' \xrightarrow{!} 1 \xrightarrow{\pi} \A(\Gamma,A)$}
        \DisplayProof
        \qquad
        \AXC{$\Gamma';\Gamma \vdash^0 MN : B$}
        \RL{App}
        \UIC{$\Gamma' \xrightarrow{\tuple{M,N}} \A(\Gamma,B^A) \times \A(\Gamma,A) \iso \A(\Gamma,B^A \times A) \xrightarrow{\A(\Gamma, \epsilon)} \A(\Gamma, B)$}
        \DisplayProof
    \end{center}
    \begin{center}
        \AXC{$\Gamma';\Gamma \vdash^0 \lambda x.M : A \to B$}
        \RL{Abs}
        \UIC{$\Gamma' \xrightarrow{M} \A(\Gamma \times A,B) \iso \A(\Gamma,B^A)$}
        \DisplayProof
        \qquad
        \AXC{$\Gamma \vdash^1 \quo M : \Box A$}
        \RL{Quo}
        \UIC{$\Gamma \xrightarrow{M} \A(1,A)$}
        \DisplayProof
        \qquad
        \AXC{$\Gamma'; \Gamma \vdash^0 \unq M : A$}
        \RL{Unq}
        \UIC{$\Gamma' \xrightarrow{M} \A(1,A) \xrightarrow{\A(!,A)} \A(\Gamma,A)$}
        \DisplayProof
    \end{center}
    \caption{Interpretation of two-level {\lambdabox} in a $\V$-enriched category $\A$ (with identification of $\den{A}$ and $A$)}
    \label{fig:twolevelinterp}
\end{figure*}

\subsection{Infinitely Enriched Categories}

We establish a mathematically precise formulation of the idea presented in the previous subsection.
The major difficulty in formalization is that we need to model infinitely many universes of the logic with infinitely iteratively enriched categories.
For calculi with finitely many levels, it suffices to construct a concrete chain of iterated enrichment in which a category of level $l$ is enriched over another category of level $l+1$.
One approach to generalize finite iteration of enrichment to infinite is to assume an infinite sequence $\tuple{C_n}_n$ of cartesian closed categories, an infinite sequence $\tuple{\A_n}_n$ of cartesian closed categories $\A_n$ enriched over $C_{n+1}$, and extra equations $\underlying{(\A_n)} = C_n$.
This definition is intuitive but too naive to be compared with other models e.g.\ Kripke categories.
Instead, we exploit the change-of-base construction of enriched categories to resolve this issue.

\begin{proposition}[Change of base]
    For monoidal categories $\V$ and $\W$, a monoidal functor $L : \V \to \W$, and a $\V$-category $\A$, $L$ induces a $\W$-category $L_{\ast}\A$.
    $L_{\ast}\A$ has the same collection of objects as $\A$, and hom objects $L_{\ast}\A(X,Y)$ given by $L\A(X, Y)$.
\end{proposition}

For a symmetric monoidal category $\V$, monoidal $\V$-category and related concepts are defined as usual (see also Appendix \ref{sec:mon-enr}).
In fact, it is known that symmetric monoidal closed categories enriched over a symmetric monoidal closed category behave quite well.
Our approach exploits the following very important result on such enriched categories.

\begin{definition}
    Given a monoidal $\V$-functor $F : \A \to \B$, the \emph{comparison} morphism of $F$ is a $\V$-natural transformation $\theta : \A(I, -) \xrightarrow{F} \B(FI, F-) \xrightarrow{\B(\iota,F-)} \B(I, F-)$
    \footnote{Our definition of comparison morphism is slightly different from \cite{Lucyshyn16}, in which the comparison morphism of $F$ is defined as $\A(I, -) \xrightarrow{\theta} \V(I, F-) \iso F$}.
\end{definition}

\begin{definition}
    A monoidal $\V$-functor $F : \A \to \B$ is \emph{normal}
    \footnote{
        Normality of monoidal functors has nothing to do with \emph{normal forms} in the term calculus, nor \emph{normal} modal logics.
        This terminology dates back to the 1960s \cite{Eilenberg66}.
    } if the comparison morphism of $F$ is an isomorphism.
\end{definition}

\begin{theorem}[\cite{Lucyshyn16}]\label{thm:fund}
    Let $\V$ be a symmetric monoidal closed category.
    The following data are equivalent up to isomorphism.
    \begin{enumerate}
        \item A symmetric monoidal closed $\V$-category $\A$.
        \item A symmetric monoidal closed (ordinary) category $\A$ and a normal symmetric monoidal functor $L : \A \to \V$.
    \end{enumerate}
\end{theorem}

\begin{proof}
    We use change of base in the proof of the upward direction.
    See the sketch in Appendix \ref{sec:formal}.
\end{proof}

This result naturally allows us to extend the definition of finitely iterated enrichment to infinite.
Intuitively, an infinitely enriched category is a (co-)limit of the chain of such iterated enrichment.
\[
    \xymatrix@R-=0.5cm{
        \ar@{}[d]_{\cdots} & \VV_3 \ar@{|->}[d] & \VV_2 \ar@{|->}[d] & \VV_1 \ar@{|->}[d] & \VV_0 \ar@{|->}[d]^{\text{underlying cat.}} \\
        \ar@{~>}[ur] & \underlying{(\VV_3)} \ar@{~>}[ur] & \underlying{(\VV_2)} \ar@{~>}[ur] & \underlying{(\VV_1)} \ar@{~>}[ur]^(0.6){\text{enrich}} & \underlying{(\VV_0)} \\
    }
\]

\begin{definition}
    An \emph{infinitely enriched category} $\VV$ is a functor from $\omega$ to the category $\NorSMCCat$ of all symmetric monoidal closed categories and symmetric monoidal normal functors.
    \[
        \xymatrix{\cdots & \VV_3 \ar[l]_{\Box_3} & \VV_2 \ar[l]_{\Box_2} & \VV_1 \ar[l]_{\Box_1} & \VV_0 \ar[l]_{\Box_0}}%
    \]
\end{definition}

\begin{remark}
    Readers may read the definition above coinductively: ``an \emph{infinitely enriched category} is a symmetric monoidal closed category enriched over another infinitely enriched category.''
\end{remark}

We informally say that an infinitely enriched category is a \emph{finitely enriched category} if there exists a natural number $n \in \N$ such that the enrichment structures above the $n$-th enrichment are all trivial (i.e., self-enrichment by the identity functor).
Specifically, say if $n = 2$, we may call them \emph{doubly enriched categories}.
In our terminology, a finitely enriched category is always an infinitely enriched category.

\begin{example}
    Every symmetric monoidal closed (ordinary) category is an infinitely enriched category by self-enrichment.
\end{example}

\begin{remark}
    For any infinitely enriched category $\VV : \omega \to \NorSMCCat$, if $\VV$ has the colimit $\E$, $\E$ works as the ``total'' category that enriches all components $\VV_n$.
    \vspace{10pt}
    \[
        \xymatrix{
            \E \ar@/^0.5pc/@{~>}[rr] \ar@/^1pc/@{~>}[rrr] \ar@/^1.5pc/@{~>}[rrrr] & \cdots & \VV_2 \ar[l] & \VV_1 \ar[l] & \VV_0 \ar[l]
        }
    \]
\end{remark}

To be used for a model for {\lambdabox}, the monoidal structures have to be defined by cartesian products.
In the following, we focus on the case $\V$ is cartesian monoidal.

\begin{definition}
    A cartesian $\V$-category is a $\V$-category $\A$ equipped with $\V$-adjunctions $\Delta \dashv \times$ and $! \dashv 1$, where $\Delta$ is the diagonal 1-cell $\A \to \A \times \A$ in $\VCat$ and $! : \A \to 1$ is the unique 1-cell into the terminal $\V$-category.
\end{definition}

\begin{lemma}
    A cartesian $\V$-category is a symmetric monoidal $\V$-category.
\end{lemma}

\begin{lemma}
    Given a cartesian $\V$-category $\A$, $\underlying{\A}$ is a cartesian category.
    Moreover, $\underlying{\A(A,-)} : \underlying{\A} \to \V$ is a cartesian functor.
\end{lemma}

\begin{definition}
    An infinitely enriched category $\VV$ is \emph{cartesian} if $\VV_n$ is a cartesian closed category and $\VV_n \to \VV_{n+1}$ is a cartesian functor for each $n$.
\end{definition}

\begin{example}
    Every cartesian closed (ordinary) category is a cartesian infinitely enriched category.
\end{example}

\subsection{Semantics}

Assume that {\lambdabox} has the product and unit types and related $\beta\eta$ rules.

\begin{definition}[{\lambdabox}-theory]
    We define an equational theory $\Delta \vdash^l M = N : A$ for {\lambdabox}-terms $\Delta \vdash^l M : A$ and $\Delta \vdash^l N : A$ by the closure of the $\beta$ and $\eta$ rules.
\end{definition}

\begin{definition}[Interpretation]
    Let $\CC$ be a cartesian infinitely enriched category.
    Interpretation $\den{-}$ of types and contexts is straightforward (see Section \ref{sec:idea}).
    For any types $A^l$ and contexts $\Gamma^l$, both $\den{A^l}$ and $\den{\Gamma^l}$ are objects in $\CC_l$.
    Given a context stack $\Delta^l$, associate an auxiliary finite-product-preserving functor $\Delta^l(-) : \CC_l \to \CC_{l+\#\Delta}$.
    ($\epsilon$ denotes the empty context stack, and $[-,-]$ denotes exponentiation.)
    It may be helpful for the definition to recall that $\Delta$ is used for a context stack in this paper, whereas $\Gamma$ is for a single context.
    \[
        \epsilon(X) = X
        \qquad
        \Gamma(X) = \Box[\den{\Gamma}, X]
        \qquad
        (\Gamma; \Delta)(X) = \Box[\den{\Gamma}, \Delta(X)]
    \]
    Then terms are interpreted as morphisms of the form $\den{\Gamma; \Delta \vdash^l M : A} : \den{\Gamma} \to \Delta(\den{A})$ by induction on the structure of the judgment and the height of its context stack.
    \begin{align*}
        &\den{\Gamma \vdash^l x : A} = \begin{aligned} \den{\Gamma} \xrightarrow{\pi} \den{A} \end{aligned}
        \qquad
        \den{\Gamma; \Delta; \Gamma' \vdash^l x : A} = \begin{aligned}
            \den{\Gamma} \xrightarrow{!} 1 \iso \Box 1 \xrightarrow{\widetilde{\Box \den{\Delta; \Gamma' \vdash^l x : A}}} (\Delta; \Gamma')(\den{A})
        \end{aligned} \\
        &\den{\Gamma \vdash^l \lambda x.M : A \to B} = \begin{aligned}
            \den{\Gamma} \xrightarrow{\widetilde{\den{M}}} [\den{A},\den{B}]
        \end{aligned} \\
        &\den{\Gamma; \Delta; \Gamma' \vdash^l \lambda x.M : A \to B} = \begin{aligned}
            \den{\Gamma} \xrightarrow{\den{M}} \Delta(\Box[\den{\Gamma',A},\den{B}]) \iso \Delta(\Box[\den{\Gamma'}, [\den{A},\den{B}]])
        \end{aligned} \\
        &\den{\Gamma; \Delta \vdash^l M N : B} = \begin{aligned}
            \den{\Gamma} \xrightarrow{\tuple{\den{M},\den{N}}} \Delta([\den{A},\den{B}]) \times \Delta(\den{A}) \iso \Delta([\den{A}, \den{B}] \times \den{A}) \xrightarrow{\Delta(\ev)} \Delta(\den{B})
        \end{aligned} \\
        &\den{\Gamma; \Delta \vdash^l \quo M : \Box A} = \begin{aligned}
            \den{\Gamma} \xrightarrow{\den{M}} \Delta(\Box[1,\den{A}]) \iso \Delta(\Box\den{A})
        \end{aligned} \\
        &\den{\Gamma; \Delta; \Gamma' \vdash^l \unq M : A} = \begin{aligned}
            \den{\Gamma} \xrightarrow{\den{M}} \Delta(\Box\den{A}) \to \Delta(\Box[\den{\Gamma'},\den{A}])
        \end{aligned}
    \end{align*}
\end{definition}

The above interpretation slightly differs from that presented in Figure~\ref{fig:twolevelinterp} so that the type former $\Box$ coincides with the normal functor $\Box$.
As a result, it identifies two judgments $\Delta; \Gamma \vdash^l M : A$ and $\Delta \vdash^{l+1} \quo\lambda\Gamma.M : \Box(\rg(\Gamma) \to A)$, where $\rg(x_1:A_1,\cdots,x_n:A_n) \to A$ means $A_1 \to \cdots \to A_n \to A$.
We will introduce in Section~\ref{sec:contextual} another interpretation which makes $M$ and $\quo M$ precisely the same morphism.

\begin{theorem}[Soundness and completeness]
    $\Delta \vdash^l M = N : A$ if and only if $\den{M} = \den{N}$ holds for any cartesian infinitely enriched category.
\end{theorem}

\begin{proof}
    Completeness is proved by the standard term model construction.
    Construct cartesian closed categories $\CC_l$ for all $l$ by collecting all terms (with one free variable) of the same level $\CC_l(A,B) = \{ x : A \vdash^l M : B \}$.
    Normal functors $\Box_l : \CC_l \to \CC_{l+1}$ arise from $(x : A \vdash^l M : B) \mapsto (y : \Box A \vdash^{l+1} \quo [(\unq y)/x]M : \Box B)$.
\end{proof}

Immediately from the proof of the completeness,  we can see that {\lambdabox} theories form internal languages (a la Lambek and Scott) of cartesian infinitely enriched categories.
Given a cartesian infinitely enriched category $\CC$, there exists a {\lambdabox} theory such that its term model is (naturally) isomorphic to $\CC$.
We call such a {\lambdabox} theory an \emph{internal language} of $\CC$.

\begin{remark}
    Normality of monoidal functors corresponds to the necessitation and denecessitation in logic.
    \[
        \AXC{$1 \to X$ in enriched cat. $\CC_l$}
        \doubleLine
        \UIC{$1 \to \CC_l(1, X)$ in enriching cat. $\CC_{l+1}$}
        \DisplayProof
        =
        \AXC{$\vdash^l X$ in object}
        \doubleLine
        \UIC{$\vdash^{l+1} \Box X$ in meta}
        \DisplayProof
    \]
\end{remark}

Of course, (not necessarily cartesian) infinitely enriched categories provide semantics for multi-staged linear lambda calculus.
A doubly enriched model of a linear calculus is discussed later.

\section{Finite Approximate Constructions}\label{sec:constructions}

This section gives two constructions of infinitely enriched categories.
Both constructions are finite approximate constructions: they can generate finitely enriched categories of any length from one structure, but cannot produce one with infinitely many non-trivial enrichment.
We remark that finite approximation is sufficient for model construction of the type theory since any context stack cannot be infinitely high.
We leave finding generic construction methods of truly infinitely enriched categories as an important future work towards the theory of iterated enrichment.

\subsection{Iterative Change-of-base Construction}

A Kripke category canonically induces a finite approximation of a cartesian infinitely enriched category by iterating the change-of-base construction.
The following is a fundamental property of symmetric monoidal closed categories.

\begin{lemma}\label{lem:compar-symmon}
    Given symmetric monoidal closed categories $\V$ and $\W$ and a symmetric monoidal functor $F : \V \to \W$, the comparison morphism $\theta$ of $F$ is a symmetric monoidal natural transformation.
\end{lemma}

\begin{proof}
    Monoidality of $\theta$ is checked by straightforward element-wise calculation.
\end{proof}

\begin{lemma}\label{lem:cob-symmon-nat}
    Given a symmetric monoidal natural transformation $\phi : G \to F : \V \to \W$, the component ${\phi_\ast}_\A : F_\ast{\A} \to G_\ast{\A}$ of $\phi_\ast$ is (strict) symmetric monoidal, if $\A$ is a symmetric monoidal $\V$-category.
\end{lemma}

\begin{corollary}
    Given a symmetric monoidal functor $F : \V \to \W$ between symmetric monoidal categories and a symmetric monoidal $\V$-category $\A$, $F$ induces a symmetric monoidal functor $\overline{F} : \underlying{\A} \to \underlying{(F_{\ast}\A)}$.
    Moreover, $\overline{F}$ is an isomorphism if $F$ is normal.
\end{corollary}

\begin{proof}
    Applying $\theta$ in Lemma \ref{lem:compar-symmon} to ${(-)}_\ast$, we obtain the 2-cell $\theta_\ast : \underlying{(-)} \to \underlying{(F_\ast(-))} : \VCat \to \WCat$.
    By Lemma \ref{lem:cob-symmon-nat}, there is a symmetric monoidal functor $\overline{F} = {\theta_{\ast}}_\A$.
\end{proof}

\begin{remark}
    $\overline{F}$ sends a morphism $f : X \to Y = I \to \A(X,Y)$ to $\overline{F}f = X \to Y = I \xrightarrow{\iota} FI \xrightarrow{Ff} F\A(X,Y)$.
\end{remark}

Let $\V$ be a Kripke category with a monoidal endofunctor $\Box : \V \to \V$.
Using the above lemma and corollaries, we can iteratively perform the change-of-base construction starting from $\V$.
\[
    \xymatrix@R-=0.5cm{
        & \Box_{\ast}\V \ar@{|->}[d] & \overline{\Box}_{\ast}{(\Box_{\ast}\V)} \ar@{|->}[d] & \overline{\overline{\Box}}_{\ast}{(\overline{\Box}_{\ast}{(\Box_{\ast}\V)})} \ar@{|->}[d] &\\
        \V \ar@(l,u)[]^{\Box} \ar@{~>}[ur] \ar[r]_-{\overline{\Box}} & \underlying{(\Box_{\ast}\V)} \ar@{~>}[ur] \ar[r]_-{\overline{\overline{\Box}}} & \underlying{(\overline{\Box}_{\ast}{(\Box_{\ast}\V)})} \ar@{~>}[ur] \ar[r]_-{\overline{\overline{\overline{\Box}}}} & \underlying{(\overline{\overline{\Box}}_{\ast}{(\overline{\Box}_{\ast}{(\Box_{\ast}\V)})})} \ar@{}[ur]_{\cdots}
    }
\]
Note that $\overline{\Box}_{\ast}{(\Box_{\ast}\V)} = {(\overline{\Box} \after \Box)}_{\ast}\V$, etc.
The hom sets of the constructed category say at level 2 are explicitly written down as follows.
\[
    \underlying{(\overline\Box_\ast(\Box_\ast{\V}))}(X,Y) = \underlying{(\Box_\ast{\V})}(I, \overline{\Box}\after\Box[X,Y])
    = \underlying{(\Box_\ast{\V})}(I, \Box[X,Y])
    = \V(I, \Box[I, \Box[X, Y]])
\]
This corresponds to the equivalence of derivability between judgments $X \vdash Y$ and $\vdash \Box(1 \to \Box(X \to Y))$.
Logically, ${\overline{\Box}}^{(n)}$ performs necessitation for given judgments and sends them to higher stages.

The model construction in this subsection explains why a Fitch-style modal logic can be interpreted in a Gentzen-style modal logic.
It is not so difficult to define the syntactic translation of {\lambdabox} into the Gentzen-style calculus based on this construction.

\subsection{Iterative Co-Kleisli Construction}

Another construction arises from comonads.
The key observation is that a comonad over a cartesian closed category produces another cartesian closed category by the co-Kleisli construction.

\begin{lemma}
    Given a comonad $\tuple{T, \delta, \epsilon}$ over a cartesian closed category $\V$, the co-Kleisli category $\V_T$ is again cartesian closed, if $T$ preserves binary products
    \footnote{More generally, any oplax monoidal comonad over a symmetric monoidal closed category produces another symmetric monoidal closed category via co-Kleisli construction.}.
\end{lemma}

Not surprisingly, this construction does not yield an infinitely enriched category at once.
First, we develop the theory of Kleisli constructions, in order to obtain infinitely many cartesian closed categories from one.
A similar construction is explored in \cite{Cheng11}.
Let $\V$ be a category and $\tuple{T, \mu^T, \eta^T}$ and $\tuple{S, \mu^S, \eta^S}$ monads over $\V$.

\begin{definition}
    A \emph{distributive law} of $T$ over $S$ is a natural transformation $l : ST \to TS$ subject to the following conditions.
    \begin{align*}
        \mu^T S \after Tl \after lT &= l \after S \mu^T & l \after S \eta^T &= \eta^T S &
        T \mu^S \after lS \after Sl &= l \after \mu^S T & l \after \eta^S T &= T \eta^S
    \end{align*}
\end{definition}

Let $l : ST \to TS$ be a distributive law of $T$ over $S$.

\begin{lemma}
    $\tuple{TS, \mu^{TS}, \eta^{TS}}$ is a monad over $\V$, where $\mu^{TS}$ and $\eta^{TS}$ are given by $\mu^{TS} = TSTS \xrightarrow{TlS} TTSS \xrightarrow{\mu^T\mu^S} TS$ and $\eta^{TS} = 1 \xrightarrow{\eta^T\eta^S} TS$.
\end{lemma}

\begin{definition}
    $l$ induces an endofunctor $S_T : \V_T \to \V_T$ over the Kleisli category, mapping $S_T(A) = SA$ and $S_T(A \xrightarrow{f} TB) = SA \xrightarrow{Sf} STB \xrightarrow{l} TSB$.
    This endofunctor $S_T$ is called the \emph{Kleisli lifting} of $S$ along $T$.
\end{definition}

\begin{lemma}
    There exist natural transformations $\eta^S_T : 1 \to S_T$ and $\mu^S_T : {S_T}^2 \to {S_T}$ such that $\tuple{S_T, \mu^S_T, \eta^S_T}$ forms a monad over the Kleisli category $\V_T$.
\end{lemma}

\begin{lemma}
    The Kleisli categories $\V_{TS}$ and $(\V_T)_{S_T}$ are identical.
    Moreover, their adjunctions commute.
    \[
        \xymatrix{C \rtwocell{'\rotatebox{90}{$\vdash$}} & C_T \rtwocell{'\rotatebox{90}{$\vdash$}} & (C_T)_{S_T}} = \xymatrix{C \rtwocell{'\rotatebox{90}{$\vdash$}} & C_{TS}}
    \]
\end{lemma}

In what follows we only consider the case where $\tuple{T, \mu^T, \eta^T}$ and $\tuple{S, \mu^S, \eta^S}$ are the same.
We then want to lift a distributive law $l : T^2 \to T^2$ along $T$ to the Kleisli category in the same way as $\mu^S_T$ and $\eta^S_T$.
To do this, we introduce another axiom for distributive laws.

\begin{definition}
    A distributive law $l : T^2 \to T^2$ is \emph{self-distributive} if it is subject to the \emph{Yang-Baxter equation}:
    \[
        Tl \after lT \after Tl = lT \after Tl \after lT.
    \]
    We say $l$ is a \emph{self-distributive law} if it is self-distributive.
\end{definition}

Self-distributivity not only allows Kleisli lifting of distributive laws but also makes the lifted distributive laws again self-distributive.

\begin{lemma}
    Let $l : T^2 \to T^2$ be a self-distributive law of $T$.
    $l$ induces another self-distributive law $l_T$ of $T_T$, by $l_T = \eta \after l : T^2 \to T^3$.
\end{lemma}

We summarize the results above in terms of comonads.

\begin{corollary}
    Assume given a comonad $\tuple{T, \delta, \epsilon}$ over $\V$, and a self-distributive law $l : T^2 \to T^2$.
    By the duals of the preceding lemmas, there are a comonad $\tuple{T_T, \delta_T, \epsilon_T}$ over $\V_T$ and a self-distributive law $l_T : {T_T}^2 \to {T_T}^2$ given by the following data in $\V$:
    \begin{align*}
        T_T(TA \xrightarrow{f} B) &= T^2A \xrightarrow{l} T^2A \xrightarrow{Tf} TB &
        \delta_T &= T^2 \xrightarrow{\epsilon} T \xrightarrow{\delta} T^2 &
        \epsilon_T &= T^2 \xrightarrow{T\epsilon} T \xrightarrow{\epsilon} 1 &
        l_T &= T^3 \xrightarrow{\epsilon} T^2 \xrightarrow{l} T^2
    \end{align*}
\end{corollary}

We also check that co-Kleisli lifting preserves monoidality.

\begin{lemma}
    Let $\tuple{T, \delta, \epsilon}$ and $l$ as the preceding corollary.
    If $T$ preserves binary products, so is $T_T$.
    Moreover, if $T$ is finite-product-preserving, so is $T_T$.
\end{lemma}

Next, we make $\V_{T^{n+1}}$ enriched over $\V_{T^n}$.
One may think assuming $T$ is internal is enough, but then $T_T$ may not be internal.
We instead take another approach, assuming all $T$-algebras.

\begin{definition}
    Let $\tuple{T, \mu, \eta}$ be a monad over $\V$.
    A \emph{$T$-algebra} is a morphism $\alpha_X : TX \to X$ in $\V$ subject to $\alpha_X \after T{\alpha_X} = \alpha_X \after \mu_X$ and $\alpha_X \after \eta_X = 1$.
\end{definition}

\begin{definition}
    $\tuple{\V, T, \mu, \eta, \alpha}$ is a \emph{category with all $T$-algebras} if $\tuple{T, \mu, \eta}$ is a monad over $\V$ and $\alpha : T \to 1$ is a natural transformation such that for each $X$ in $\V$, $\alpha_X$ is a $T$-algebra.
\end{definition}

Under a certain reasonable assumption, $T$-algebras can be lifted along $T$ with a self-distributive law.

\begin{definition}
    $T$-algebra $\alpha$ is \emph{self-distributive} if $l \after T\alpha = \alpha T$.
\end{definition}

\begin{proposition}
    If $\alpha$ is self-distributive, morphisms $\alpha_T : T_T \to 1$ in $\V_T$ given by $\eta \after \alpha : T \to T$ form a natural transformation of self-distributive $T_T$-algebras.
\end{proposition}

Now, we are ready to construct (cartesian) infinitely enriched categories.
We have seen that a finite-product-preserving self-distributive comonad $T$ over a cartesian closed category with all self-distributive $T$-coalgebras induces another such by Kleisli lifting.
This operation can be repeated arbitrarily for finitely many times.
By the next lemma we can regard $\V_T$ as a $\V$-category, so iterating the operation produces a (cartesian) infinitely enriched category.

\begin{lemma}
    Assume given a cartesian closed category with all $T$-coalgebras $\tuple{\V, T, \delta, \epsilon, \alpha}$, and $T$ preserves finite products.
    Then a normal cartesian functor $F : \V_T \to \V$ is given by $F(X) = X$ and $F(f : X \to Y) = X \xrightarrow{\alpha_X} TX \xrightarrow{f} Y$.
\end{lemma}

\begin{proof}
    Normality follows from $T1 \iso 1$ and naturality of $\alpha$.
    Use the fact that $\delta_{1} = \alpha_{T1} = {!} : T1 \to T^2{1}$.
\end{proof}

\begin{proposition}
    $\V_{T^n}$ forms a cartesian infinitely enriched category, with $\V_{T^{n-k}}$ as the $k$-th enriching category.
\end{proposition}

The theory of such comonads has a non-trivial model.

\begin{example}
    Letting $A$ be a monoid object, ${(-)}^A$ forms a comonad that preserves finite products.
    Then the morphism $! : A \to 1$ gives rise to coalgebras $X \to X^A$, and a self-distributive law is given by the swapping $\sigma = \tuple{\pi_2,\pi_1}$.
    The additive monoid over $\N$ in $\Set$ is a typical example, which is a model for the comonadic framework of stream programming.
\end{example}

\section{Modal Axioms and Effects}\label{sec:models}

In this section, we show and discuss various instances of the change-of-base semantics.

\subsection{LTL Next Modality}

Full and faithful cartesian closed functors are known to give a complete semantics for {\lambdacircle}, which corresponds to the ``next'' fragment of intuitionistic linear temporal logic (LTL) in the sense of Curry-Howard~\cite{DBLP:journals/jacm/Davies17,Benaissa98}.
However, the requirement for full faithfulness seems weird to us because the next operator is known to be characterized by the following additional axiom scheme \cite{DBLP:journals/iandc/KojimaI11}:
\[
    (\Circ A \to \Circ B) \to \Circ(A \to B),
\]
which seems to assert closedness of the functor $\Circ$, not requiring the full faithfulness.
The following theorem reveals that full faithfulness, in fact, requires normality of the cartesian closed functor.

\begin{theorem}
    Assume given a cartesian monoidal functor $\Circ : \V \to \W$ between cartesian closed categories.
    $\Circ$ is full and faithful if and only if $\Circ$ is normal and strong closed, i.e., $\widetilde{\Circ\ev} : \Circ[X,Y] \to [\Circ X,\Circ Y]$ is an isomorphism.
\end{theorem}

\begin{lemma}
    The Yoneda embedding $y : \V \to \PSh{\V}$ is a normal cartesian closed functor.
\end{lemma}

Therefore, every cartesian closed category may be regarded as a category enriched over the category of presheaves.

\begin{corollary}
    Given a cartesian closed category $\V$, $y$ gives rise to a cartesian doubly enriched category $\xymatrix@C-=0.5cm{\V & \PSh{\V} \ar@{~>}[l]}$.
\end{corollary}

\begin{remark}
    Cartesian doubly enriched categories arising from the Yoneda embedding are also used by Hofmann \cite{DBLP:conf/lics/Hofmann99} to offer categorical semantics of higher-order abstract syntax (HOAS), and by Bekki and Masuko \cite{DBLP:conf/jsai/Bekki08,Bekki14} to study semantics of natural languages.
    The enriched category-theoretic viewpoint could lead us to further analysis of the underlying connections between modal logics and domain-specific languages.
\end{remark}

\subsection{S4 Modality and Linear Non-linear Models}\label{sec:linearnonlinear}

S4 is a normal modal logic characterized by the following axiom schemata: $\Box A \to A$ and $\Box A \to \Box\Box A$.
Almost clear from the axioms, S4 modality is categorically characterized by (monoidal) comonads\cite{Bierman96}.
As well as Kripke categories, we can perform iterative change-of-base construction along the underlying monoidal functor of a monoidal comonad.

\begin{lemma}
    Given a monoidal comonad $\Box$ over a cartesian closed category $\V$, there are $\V$-functors $\epsilon_n : \Box^{n+1}_{\ast}\V \to \Box^{n}_{\ast}\V$ and $\delta_n : \Box^{n+1}_{\ast}\V \to \Box^{n+2}_{\ast}\V$ for each $n$, subject to
    $\delta_{n+1}\after\delta_n = \Box_\ast{\delta_n}\after\delta_n$ and $\epsilon_{n+1}\after\delta_n = 1 = \Box_\ast{\epsilon_n}\after\delta_n$.
    The following picture illustrates how $\epsilon_n$ arises.
    \[
        \xymatrix@R-=0.5cm{
            \V \dtwocell^{1}_{\Box}{^\epsilon} \ar@{|->}[r]^-{{(-)}_\ast} & \SMCVCat \dtwocell^{1}_{\Box_\ast}{^\epsilon_{\ast}} \ar@{}[r]|{\ni} & \Box^{n}_\ast{\V} \ar@{|->}[d]_{\Box_\ast} \ar@{|->}[dr]^{1}\\
            \V \ar@{|->}[r]^-{{(-)}_\ast} & \SMCVCat \ar@{}[r]|{\ni} & \Box^{n+1}_\ast{\V} \ar[r]_{\epsilon_n := {\epsilon_\ast}_{\Box^{n}_\ast{\V}}} & \Box^{n}_\ast{\V}
        }
    \]
\end{lemma}

There is, however, a more convenient way to characterize a monoidal comonad in terms of enriched categories.
Like ordinary adjunctions, any monoidal comonad arises from a monoidal adjunction.

\begin{lemma}
    Given a monoidal adjunction $F \dashv G : \V \to \W$, $F$ is strong monoidal and $G$ is normal.
\end{lemma}

\begin{proof}
    The natural isomorphism $\V(X,GY) \iso \W(FX,Y)$ sends the unit $\iota_G$ of $G$ to the inverse of the unit $\iota_F$ of $F$.
    Normality of $G$ follows from $I \to X \sslash FI \xrightarrow{\iota_F^{-1}} I \to X \sslash I \xrightarrow{\iota_G} GI \to GX$.
\end{proof}

Therefore, given a (symmetric) monoidal comonad $\Box : \V \to \V$ whose co-Kleisli category is symmetric monoidal closed, we can regard it as a doubly enriched category $\xymatrix@C-=0.5cm{\V & \V_\Box \ar@{~>}[l]}$.

In fact, this phenomenon was perceived at least in the 1990s in the studies of linear logic\cite{DBLP:conf/csl/Benton94}.
Symmetric monoidal closed categories with a monoidal adjunction to a cartesian closed category are called linear non-linear models (LNL models) \cite{DBLP:conf/tlca/Bierman95,Schalk04}, and known to form a categorical model of the IMELL.
Benton proposed in their pioneering paper a logic and its corresponding term calculus that separate linear and classical inference into two distinct types of judgments.
The syntax and semantics can be generalized to other adjunctions, and the system is sometimes called \emph{adjoint logic}~\cite{DBLP:conf/lics/BentonW96}.
In Benton's calculus, linear and classical judgments interact with each other via the following four rules.
The symbols $\mathtt{F}$ and $\mathtt{G}$ are the syntactic counterpart of the adjoint pair $F$ and $G$ above.
\[
    \AXC{$\Gamma \vdash_{\mathcal{C}} s : X$}
    \UIC{$\Gamma; \cdot \vdash_{\mathcal{L}} \mathtt{F}(s) : \mathtt{F}X$}
    \DisplayProof
    \quad
    \AXC{$\Gamma; \Theta \vdash_{\mathcal{L}} e : \mathtt{F}X$}
    \AXC{$\Gamma, x:X; \Xi \vdash_{\mathcal{L}} f : A$}
    \BIC{$\Gamma; \Theta, \Xi \vdash_{\mathcal{L}} \mathop{\mathtt{let}}\mathtt{F}(x) = e \mathrel{\mathtt{in}} f : A$}
    \DisplayProof
    \quad
    \AXC{$\Gamma; \cdot \vdash_{\mathcal{L}} e : A$}
    \UIC{$\Gamma \vdash_{\mathcal{C}} \mathtt{G}(e) : \mathtt{G}A$}
    \DisplayProof
    \quad
    \AXC{$\Gamma \vdash_{\mathcal{C}} s : \mathtt{G}A$}
    \UIC{$\Gamma; \cdot \vdash_{\mathcal{L}} \mathtt{derelict}(s) : A$}
    \DisplayProof
\]

Although $\mathcal{L}$ is a SMCC in original Benton's calculus, it is straightforward to modify it so that $\mathcal{L}$ be a CCC.
Regarding $\vdash_{\mathcal{L}}$ as $\vdash^0$ and $\vdash_{\mathcal{C}}$ as $\vdash^1$, readers notice that the latter two rules are exactly quotation and unquotation.
Therefore, this calculus can be considered an extension of two-level {\lambdabox}.

On the other hand, the first two rules are very much like the dual-context calculus presented in Section~\ref{sec:other}.
Indeed, it is possible to introduce levels to the dual-context system with the same convention as the latter two (i.e., $\vdash_{\mathcal{C}}$ as the meta level of $\vdash_{\mathcal{L}}$).
\[
    \AXC{$\cdot \mid \Gamma \vdash^{l+1} A$}
    \RL{$\Box$-I}
    \UIC{$\Gamma \mid \Gamma' \vdash^{l} \Box A$}
    \DisplayProof
    \qquad
    \AXC{$\Gamma \mid \Gamma' \vdash^{l} \Box A$}
    \AXC{$\Gamma, A \mid \Gamma' \vdash^{l} B$}
    \RL{$\Box$-E}
    \BIC{$\Gamma \mid \Gamma' \vdash^{l} B$}
    \DisplayProof
\]
While boxing of Fitch-style increases the level of a proof, that of dual-context \emph{decreases} it.
It can be seen easily that assignment of levels does not affect the provability.
If we state it formally, the dual-context version of Theorem~\ref{thm:labeling} holds.

By this labeling, we can conclude that Benton's calculus is nothing but a fusion of Fitch-style and dual-context systems restricted to two-levels, where $\mathtt{F}$ is the dual-context box and $\mathtt{G}$ is the Fitch-style box.
As a consequence, one can see that the box modality of the dual-context system is the left adjoint of the box modality of the Fitch-style system.\footnote{
    Such an investigation makes sense as far as neither disjunctions nor bottom is dealt with.
    Since disjunctions are usually interpreted as colimits, a left adjoint $F$ has to preserve disjunctions.
    In this paper, we focus on only box modalities and conjunctions.
}

It is not difficult to extend Benton's calculus to multi-levels.
As we have done in {\lambdabox}, we can just replace outer contexts $\Gamma$ with context stacks $\Delta$.
\[
    \AXC{$\Delta \vdash^{l+1} A$}
    \RL{$\Box$-I}
    \UIC{$\Delta; \Gamma' \vdash^{l} \Box A$}
    \DisplayProof
    \qquad
    \AXC{$\Delta; \Gamma' \vdash^{l} \Box A$}
    \AXC{$\Delta, A; \Gamma' \vdash^{l} B$}
    \RL{$\Box$-E}
    \BIC{$\Delta; \Gamma' \vdash^{l} B$}
    \DisplayProof
\]
These rules are a multi-context version of the dual-context system, which we shall call \emph{multi-context} system.
Although this multi-context system is equivalent to the dual-context system with respect to the provability, types and boxes are strictly categorized into levels.
The multi-level Benton-style calculus can be defined as the union of the multi-context and Fitch-style systems.
A model of the multi-level Benton-style calculus is given as an infinite sequence of monoidal adjunctions, which forms an infinitely enriched category.
\[
    \xymatrix{\cdots \ar@/^/[r] \ar@{}[r]|{\rotatebox{90}{$\vdash$}} & \VV_3 \ar@/^/[r] \ar@{}[r]|{\rotatebox{90}{$\vdash$}} \ar@/^/[l] & \VV_2 \ar@/^/[r] \ar@{}[r]|{\rotatebox{90}{$\vdash$}} \ar@/^/[l] & \VV_1 \ar@/^/[r] \ar@{}[r]|{\rotatebox{90}{$\vdash$}} \ar@/^/[l] & \VV_0 \ar@/^/[l]}
\]
In this sense, the Fitch-style box and the dual-context box are an adjoint pair in the multi-level system.

In Benton's semantics, a $0$-level judgment $\den{\Gamma; \Gamma' \vdash^{0} M : A}$ is interpreted as a morphism $F\den{\Gamma} \times \den{\Gamma'} \to \den{A}$ in $\VV_0$.
If we regard $\VV_0$ as an enriched category through $G$, however, it can be considered a morphism from $\den{\Gamma} \to G[\den{\Gamma'}, \den{A}]$ in $\VV_1$.
The latter interpretation is just the semantics proposed in the previous section.
So, the existence of the left adjoint enables us to interpret the Fitch-style system in a lower level category without enrichment.
This idea has already been proposed by Clouston in \cite{DBLP:journals/corr/abs-1710-08326}, although the monoidal condition of $F$ is relaxed.
In the paper, Clouston gave a sound categorical interpretation for Fitch-style type theory via a diamond modality, which is a left adjoint of the box.
In this sense, we can say that Clouston's interpretation uses a dual-context flavor for the Fitch-style calculus.
Of course, since the diamond is not monoidal, connection between the diamond and the multi-context box is less trivial when more than two contexts are involved.
For example, given a three-level context $\Gamma; \Gamma'; \Gamma''$, it is interpreted as $F(F\den{\Gamma} \times \den{\Gamma'}) \times \den{\Gamma''}$ in Clouston's model.
If $F$ is strong monoidal, it is isomorphic to $F^2\den{\Gamma} \times F\den{\Gamma'} \times \den{\Gamma''}$ and coincides with the interpretation of the multi-context system.

Although readers may worry about the fact that original Clouston's model is not a sequence of adjunctions but one adjunction on one category, it is not essential for the discussions.
If we consider the case all categories in an infinitely enriched category coincide, the box is interpreted as an endofunctor on that category.
Hence, a (normal) Kripke category is a special case of infinitely enriched categories, and our studies can be applied to a Kripke category under the assumption that $G$ is normal.
In addition, conversely, it is possible to generalize Clouston's models and interpretation with multiple levels straightforwardly.
In that case, appropriate levels should be assigned to the occurrences of $F$ in the above paragraph.

\subsection{Multi-staged Effectful Computation}

We can also extend the semantics to effectful computation.
Formal semantics of modal type theories with effects is not well studied because a Kleisli category is not a monoidal category in general.
While one of the authors has provided Gentzen-style semantics for the effectful modal type system in \cite{DBLP:conf/aplas/Kakutani07}, this paper proposes Fitch-style semantics allowing each stage to have a different effect.

\begin{definition}
    A \emph{$\V$-monad} is a lax functor from the terminal 2-category $1$ to $\VCat$.
    Similarly, \emph{monoidal $\V$-monad} and \emph{symmetric monoidal $\V$-monad} are defined as monads in $\MonVCat$ and $\SMVCat$.
\end{definition}

\begin{remark}
    Given a $\V$-monad $T$ over a $\V$-category $\A$, $T_0$ is a monad over $\A_0$.
\end{remark}

\begin{lemma}
    If $\V$ is symmetric monoidal closed, $\V$-monad is exactly the same as strong monad.
\end{lemma}

It immediately follows that every monad on $\Set$ is strong. Assume that $\V$ is symmetric and closed.

\begin{lemma}
    Given a $\V$-monad $T$ on $\V$, or equivalently given a strong monad, the Kleisli category $\V_T$ is canonically enriched over $\V$.
    Moreover, the Kleisli adjunction between $\V$ and $\V_T$ is also $\V$-enriched.
\end{lemma}

\begin{proposition}
    Given symmetric monoidal closed categories $\V$ and $\W$, a normal monoidal functor $F : \V \to \W$, and a $\V$-monad $T$ on $\V$, $\V_T$ is canonically enriched over $\W$.
\end{proposition}

The above construction gives semantics for a call-by-value multi-staged lambda calculus such that each stage has its own effect.
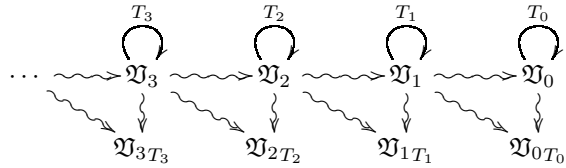
\[
    \xymatrix@R-=0.5cm{
        \ar@{}[d]{\cdots} \ar@{~>}[r] \ar@{~>}[rd] & \VV_3 \ar@(ul,ur)[]^{T_3} \ar@{~>}[r] \ar@{~>}[rd] \ar@{~>}[d] & \VV_2 \ar@{~>}[d] \ar@{~>}[rd] \ar@(ul,ur)[]^{T_2} \ar@{~>}[r] & \VV_1 \ar@{~>}[d] \ar@{~>}[rd] \ar@(ul,ur)[]^{T_1} \ar@{~>}[r] & \VV_0 \ar@{~>}[d] \ar@(ul,ur)[]^{T_0} \\
        & {\VV_3}_{T_3} & {\VV_2}_{T_2} & {\VV_1}_{T_1} & {\VV_0}_{T_0}
    }
\]
We can also make Kleisli categories monoidal to enrich some other categories.
In this case, the enrichment of a Kleisli category can be explained as change-of-base construction.

\begin{lemma}
    Given a monoidal (resp. symmetric monoidal) $\V$-monad $T$ on a (resp. symmetric monoidal) $\V$-category $\A$, the enriched Kleisli category $\A_T$ is monoidal (resp. symmetric monoidal).
\end{lemma}

\section{Contextual Modality}\label{sec:contextual}

In this section, we apply the change-of-base semantics to generalized contextual modality.
Contextual modality allows more direct interpretation of quoted terms.

\subsection{Contextual Modal Type Theories}

Contextual modal type theories are type theories internalizing hypothetical judgments.
Nanevski et al.\, introduced contextual modality~\cite{DBLP:journals/tocl/NanevskiPP08} in search of the logical foundation of meta-variables and explicit substitution.
Contextual modality is described as a generalization of modality: whereas ordinary modality asserts the proposition is true under no hypotheses, contextual modality permits assertions of propositional truth under any number of hypotheses.
We denote a contextual modal type with hypotheses $\Gamma$ and conclusion $A$ by $[\Gamma]A$.
Contextual modal types $[\cdot]A$ with no hypotheses may be written as $\Box A$.

While the original formulation in \cite{DBLP:journals/tocl/NanevskiPP08} was dual-context and based on the intuitionistic S4, there is another formulation of contextual modality on top of Fitch-style system and the intuitionistic K \cite{Murase17}.
In \cite{Murase17}, the (Quo) and (Unq) rules of {\lambdabox} are replaced with the following new rules.
\[
    \AXC{$\Delta; \Gamma \vdash^l M : B$}
    \RL{Quo}
    \UIC{$\Delta \vdash^{l+1} \quo^{\tuple{\dom(\Gamma)}} M : [\rg(\Gamma)]B$}
    \DisplayProof
    \quad
    \AXC{$\Delta \vdash^{l+1} M : [A_1,\cdots, A_n]B$}
    \AXC{$\Delta; \Gamma \vdash^l N_1 : A_1 \enskip \cdots \enskip \Delta; \Gamma \vdash^l N_n : A_n$}
    \RL{Unq}
    \BIC{$\Delta; \Gamma \vdash^l \unq_{\tuple{N_1,\cdots, N_n}} M : B$}
    \DisplayProof
\]
The calculus is called {\lambdabracket}.
Logically, contextual modalities of level $l+1$ allow us to refer meta-theoretical properties of the $l$-level theory.
Interestingly, the structural rules are reasoned in the logic itself.

\begin{theorem}
    The formalized structural rules are provable in {\lambdabracket}:
    $\vdash^{l+1} [\Gamma]A \to [\Gamma,B]A$, $\vdash^{l+1} [\Gamma,B,B]A \to [\Gamma,B]A$, $\vdash^{l+1} [\Gamma,B,B',\Gamma']A \to [\Gamma,B',B,\Gamma']A$, and $\vdash^{l+1} [\Gamma]B \to [\Gamma,B]A \to [\Gamma]A$.
\end{theorem}

We replace the rewriting rules (and hence the equations) for boxes as well.
\begin{align*}
    \unq_{\tuple{N_1,\cdots,N_n}} \quo^{\tuple{x_1,\cdots,x_n}} M &\xrightarrow{\beta_{[\cdot]}} [N_1/x_1,\cdots,N_n/x_n]M & M &\xrightarrow{\eta_{[\cdot]}} \quo^{\tuple{x_1,\cdots,x_n}} \unq_{\tuple{x_1,\cdots,x_n}} M \text{ if $M : [A_1,\cdots,A_n]B$}
\end{align*}
Here, we implicitly introduced parallel substituion for contextual modalities, which we do not define here since it is beyond the scope of the paper.
Meanwhile, we claim that the desired fundamental properties of calculus all hold in {\lambdabracket}.

\begin{theorem}
    Subject reduction, strong normalization, and the Church-Rosser property hold for the $\beta$ rules.
\end{theorem}

\subsection{Semantics}

We can interpret terms of {\lambdabracket} in {\lambdabox} via syntactic translation, where quotation with variables is replaced with abstraction followed by quotation.
\[
    \quo^{\tuple{x_1,\cdots,x_n}} M \mapsto \quo\lambda x_1,\cdots,x_n.M
\]
Therefore cartesian infinitely enriched categories also become a model of {\lambdabracket} for free.
However, it should be noted that one can find a more direct and concise interpretation of {\lambdabracket} in the same model.
As remarked in Section~\ref{sec:semantics}, the interpretation of {\lambdabox} is forced to be awkward in order to make syntactic $\Box$ and normal functors coincide.
In the language of enriched categories, it sends {\lambdabox} judgments $\Gamma; \Gamma' \vdash^l A$ to morphisms of the form $\Gamma \to \A(1, [\Gamma', A])$, instead of $\Gamma \to \A(\Gamma', A)$.
This was mainly due to the fact that there is no syntactic construct exposing the contravariant part of (enriched) hom funtors.
The presence of contextual modalities enables us to avoid such workarounds and to implement precisely the naive idea presented in Section~\ref{sec:semantics}.
We can define the interpretation by identifying contextual modalities $[-]-$ and hom functors $\A(-,-)$, which are also identified with $\Box[-,-]$ via change of base.
Consequently, terms before and after quotation become exactly the same morphism in the semantics.
\[
    \den{\Gamma; x_1 : A_1,\cdots,x_n : A_n \vdash^l M : B} = \den{\Gamma \vdash^{l+1} \quo^{\tuple{x_1,\cdots,x_n}} M : [A_1,\cdots,A_n]B} = \den{\Gamma} \xrightarrow{M} \Box[{\textstyle \prod{\den{A_i}}},\den{B}]
\]

\begin{proposition}
    Cartesian infinitely enriched categories are sound and complete w.r.t.\ {\lambdabracket} with this $\den{\cdot}$.
\end{proposition}

Although {\lambdabracket} is semantically identified with {\lambdabox}, {\lambdabracket} is superior to {\lambdabox} in some practical cases.
Because {\lambdabracket} distinguishes judgmental validity $[\Gamma]A$ and their internal representations $\Gamma \to A$, we can exploit this to avoid unnecessary object-level applications to reduce runtime overhead.
For example, a quantum circuit description language \textsc{Proto-Quipper-M} has circuit types $\mathtt{Circ}(T,U)$ to represent a unit of computation with input $T$ and output $U$ in a non-higher-order way \cite{DBLP:journals/corr/RiosS17}.
The circuit types are in fact exactly the contextual modality in {\lambdabracket}, and indeed their semantics is given in terms of the $\Set$-enrichment structure of any symmetric monoidal category.
Similarly, we can easily generalize {\lambdabox} to model non-higher-order computation by removing the implicational structure in the object-level logic.
Reversible computation is a notable example; its model, the category $\PInj$ of all sets and partial injections, has the canonical dagger symmetric traced monoidal structure but no closed structure \cite{DBLP:conf/birthday/Heunen13}.
It even makes sense to assume no structures (i.e., no products, abstractions, the unit element, nor multiple variables) in the object level at all.
With this modification, the two-level Fitch-style system for contextual modality is expected to serve with internal languages to any enriched category.

\begin{ack}
    The authors thank their colleague Hiroki Kobayashi for helpful discussions and comments to an early draft of this paper.
    They also want to acknowledge criticism and encouragement from members of SLACS 2017 and CSCAT 2018.
    This work is partially supported by JSPS KAKENHI Grant Number 18J21885.
\end{ack}


\clearpage
\appendix

\section{Monoidal Enriched Categories}\label{sec:mon-enr}

Defining higher-dimensional monoidal structures is hard work.
In \cite{Day97}, a definition of \emph{monoidal objects} is given in terms of \emph{Gray monoids}, which are strictified monoidal bicategories.
By the coherence theorem of tricategories \cite{Gordon95}, it is possible to define monoidal objects in any monoidal 2-category as well.
For simplicity, we restrict our attention to monoidal (strict) 2-categories in the following.
Let $\V$ be a symmetric monoidal category.

\begin{definition}
    A \emph{monoidal 2-category} is a 2-category $\A$ together with a 2-functor $\tensor : \A \times \A \to \A$ and a 0-cell $I \in \A$ such that there are coherent 2-natural isomorphisms $\alpha: (X \tensor Y) \tensor Z \iso X \tensor (Y \tensor Z)$, $\lambda : I \tensor X \iso X$, and $\rho : X \tensor I \iso X$.
    A monoidal 2-category is \emph{symmetric} if it is endowed with another coherent 2-natural isomorphism $\sigma : X \tensor Y \iso Y \tensor X$.
\end{definition}

\begin{example}
    The 2-category $\VCat$ of all $\V$-categories, $\V$-functors, and $\V$-natural transformations is a symmetric monoidal 2-category.
    The symmetric monoidal structure of $\VCat$ is induced by the tensor of $\V$.
    That is, for $\V$-categories $\A, \B$, $\A \tensor \B$ consists of a collection $\ob\A \times \ob\B$ of objects and hom objects $(\A \tensor \B)(\tuple{X, Y}, \tuple{X', Y'}) = \A(X, X') \tensor \B(Y, Y')$.
\end{example}

\begin{definition}
    Let $\A$ be a monoidal 2-category.
    A \emph{monoidal object}, or \emph{pseudomonoid}, in $\A$ is a 0-cell $A \in \A$ together with 1-cells $\mu : A \tensor A \to A$ and $\eta : I \to A$ and 2-cells $\alpha, \lambda, \rho$ subject to the usual coherence conditions.
    \[%
        \xymatrix{
            A^3 \ar[d]_{\mu1} \ar[r]^{1\mu} & A^2 \ar[d]^{\mu} \\
            A^2 \ar[r]_{\mu} \ar@{=>}[ur]^{\alpha} & A
        }
        \qquad
        \xymatrix{
            A \ar@{}[rd]^{}="a" \ar[rd]^(0.3){1} \ar[d]_{1\eta} \ar[r]^{\eta1} & A^2 \ar[d]^{\mu} \ar@{=>}"a"^{\rho} \\
            A^2 \ar@{=>}"a"^{\lambda} \ar[r]_{\mu} & A
        }
    \]
    Here we omitted tensor symbols and coherent natural isomorphisms of $\A$.
\end{definition}

A monoidal object in $\A$ is \emph{symmetric} if it has a 2-cell $\sigma : \mu \to \mu \after \sigma_\A$ satisfying the coherent condition of symmetry.

\begin{definition}
    A \emph{monoidal $\V$-category} is a monoidal object in $\VCat$.
    Similarly, we call a symmetric monoidal object in $\VCat$ \emph{symmetric monoidal $\V$-category}.
\end{definition}

Given a monoidal $\V$-category $\A$, its 1-cells $\mu$ and $\eta$ just specify a $\V$-functor $\tensor : \A \tensor \A \to \A$ and an object $I \in \A$.

\begin{definition}
    A \emph{monoidal $\V$-functor} between monoidal $\V$-categories $\A, \B$ is a $\V$-functor $F : \A \to \B$ together with $\V$-natural transformations $\tau : \mu_\B \after (F \tensor F) \to F \after \mu_\A$ and $\iota : \eta_\B \to F \after \eta_\A$ subject to the usual axioms.
    A \emph{monoidal $\V$-natural transformation} between monoidal $\V$-functors $F, G : \A \to \B$ is similarly defined to be a $\V$-natural transformation between $F$ and $G$ that respects the $\V$-monoidal structure.
\end{definition}

\emph{Symmetric monoidal $\V$-functor} between symmetric monoidal $\V$-categories is similarly defined.
Likewise in ordinary category theory, namely in the case of $\V = \Set$, monoidal $\V$-natural transformations are automatically symmetric monoidal $\V$-natural transformations.  

\begin{example}
    Monoidal objects coincide with monoidal categories when the ambient category is $\Cat$ (i.e., $\Set$-enriched).
    Also, monoidal 2-categories are exactly monoidal $\Cat$-categories, that is, monoidal objects in $\twoCat$.
    \emph{Monoidal 2-functor} and \emph{monoidal 2-natural transformation} are similarly defined in terms of monoidal structures in $\twoCat$.
\end{example}

Assume $\V$ is closed.
Then $\V$ can be canonically regarded as a $\V$-category (by taking $\V(X,Y) = [X,Y]$).

\begin{lemma}[\cite{Kelly05,Lucyshyn16}]
    $\V$ is a symmetric monoidal $\V$-category.
\end{lemma}

\begin{definition}
    A symmetric monoidal $\V$-category is \emph{closed} if it has $\V$-adjunctions $X \tensor (-) \dashv [X, -]$ for each $X$.
\end{definition}

\begin{lemma}
    $\V$ is a symmetric monoidal closed $\V$-category.
\end{lemma}

\begin{lemma}
    For any monoidal $\V$-category $\A$, the covariant representable $\V$-functor $\A(I, -) : \A \to \V$ is monoidal.
    Moreover, if $\A$ is symmetric, $\A(I, -)$ is also symmetric.
\end{lemma}

\begin{example}
    The covariant representable functor $\V(I, -) : \V \to \Set$ is symmetric monoidal.
\end{example}

\begin{definition}
    $\V(I, -) : \V \to \Set$ induces via the change-of-base construction the forgetful 2-functor $\underlying{(-)} = \V(I, -)_{\ast} : \VCat \to \Cat$.
\end{definition}

\begin{lemma}
    There is a canonical isomorphism $\underlying{\V} \iso \V$.
\end{lemma}

\section{Formal Theory of Change-of-base}\label{sec:formal}

We review the results in the formal theory of change-of-base, mainly explored in \cite{Lucyshyn16} and \cite{Eilenberg66,Cruttwell08}.
For simplicity, we ignore the size issue of categories.

\begin{proposition}
    Given monoidal categories $\V$ and $\W$ and a monoidal functor $L : \V \to \W$, $L$ induces a 2-functor $L_{\ast} : \VCat \to \WCat$.
    More generally, change-of-base construction gives rise to a 2-functor ${(-)}_{\ast} : \MonCat \to \twoCat$.
\end{proposition}

\begin{theorem}
    $(-)_{\ast} : \MonCat \to \twoCat$ lifts to $(-)_{\ast} : \SMCat \to \SMtwoCat$ from the 2-category $\SMCat$ of all symmetric monoidal categories to the 2-category $\SMtwoCat$ of all symmetric monoidal 2-categories.
\end{theorem}

\begin{proposition}
    The operation sending a monoidal 2-category $\V$ to the 2-category of all monoidal objects in $\V$ forms a 2-functor $\Mon : \MontwoCat \to \twoCat$.
    Similarly, there is also a 2-functor $\SM : \SMtwoCat \to \twoCat$.
\end{proposition}

\begin{corollary}
    The composite $\SMCat \xrightarrow{(-)_{\ast}} \SMtwoCat \xrightarrow{\SM} \twoCat$ gives a 2-functor which sends a symmetric monoidal category $\V$ to the 2-category $\SMVCat$ of symmetric monoidal $\V$-categories.
    We also write this composite 2-functor $(-)_\ast$ by abuse of notation.
\end{corollary}

\begin{corollary}
    For any symmetric monoidal closed category $\V$ and symmetric monoidal $\V$-category $\A$, $\underlying{\A}$ is a symmetric monoidal category, and $\underlying{\A(I, -)} : \underlying{\A} \to \underlying{\V}$ is a symmetric monoidal functor.
\end{corollary}

\begin{lemma}
    $(-)_\ast : \SMCat \to \twoCat$ lifts to $(-)_\ast : \SMCCat \to \twoCat$, sending a symmetric monoidal closed category $\V$ to the full sub-2-category $\SMCVCat$ of $\SMVCat$ with all symmetric monoidal closed $\V$-categories.
\end{lemma}

\begin{corollary}
    For any symmetric monoidal closed category $\V$ and symmetric monoidal closed $\V$-category $\A$, $\underlying{\A}$ is a symmetric monoidal closed category.
\end{corollary}

\begin{proposition}
    Given a symmetric monoidal closed $\V$-category $\A$, there is a strict symmetric monoidal $\V$-isomorphism ${\underlying{\A(I, -)}}_{\ast}\underlying{\A} \iso \A$.
    Furthermore, its underlying isomorphism $\underlying{({\underlying{\A(I, -)}}_{\ast}\underlying{\A})} \iso \underlying{\A}$ is the canonical isomorphism.
\end{proposition}

\begin{lemma}\label{lem:down}
    For any $\V$-category $\A$, $\underlying{\A(I, -)} : \underlying{\A} \to \underlying{\V}$ is normal.
\end{lemma}

\begin{proposition}
    Let $L : \V \to \W$ be a symmetric monoidal functor between monoidal categories $\V$ and $\W$.
    If $L$ is normal, there exists an isomorphism $\underlying{(L_{\ast}\V)} \iso \V$.
\end{proposition}

\begin{lemma}\label{lem:up}
    Let $F : \underlying{(L_{\ast}\V)} \iso \V$ be the isomorphism in the above proposition.
    There is a natural isomorphism $L \after F \iso \underlying{L_\ast \V(I,-)}$.
\end{lemma}

\begin{proof*}{Proof of Theorem \ref{thm:fund} (Sketch)}
    For the downward direction, take the underlying covariant hom functor $\underlying{\A(I,-)} : \underlying{\A} \to \V$.
    Perform change of base $L_\ast \A$ for the upward direction.
    Their equivalence follows from the above proposition and lemma.
\end{proof*}

\end{document}